\documentclass[12 pt]{amsart}
\usepackage{amssymb,amsfonts,amsthm, color}
\usepackage{graphicx}
\usepackage{bm}
\newtheorem{thm}{Theorem}[section]
\newtheorem{cor}[thm]{Corollary}
\newtheorem{lem}[thm]{Lemma}

\theoremstyle{definition}

\theoremstyle{remark}

\numberwithin{equation}{section}

\begin{document}
\title[Comparing networks and reduced models]{How well do reduced models capture\\
the dynamics in models of interacting neurons ? }
\author{Yao Li}
\address{Yao Li: Department of Mathematics and Statistics, 
University of Massachusetts Amherst, USA}
\email{yaoli@math.umass.edu}

\author{Logan Chariker}
\address{Logan Chariker: Courant Institute of Mathematical
Sciences, New York University, New York, NY 10012, USA}
\email{bortkiew@gmail.com}
\thanks{LC was supported by a grant from the Swartz Foundation.}

\author{Lai-Sang Young}
\address{Lai-Sang Young: Courant Institute of Mathematical
Sciences, New York University, New York, NY 10012, USA}
\email{lsy@cims.nyu.edu}
\thanks{LSY was supported in part by NSF Grant DMS-1363161. }

\maketitle
\noindent
{\small {\bf Abstract.}  This paper introduces a class of stochastic models of
interacting neurons with emergent dynamics
similar to those seen in local cortical populations, and compares them to 
very simple reduced models driven by the same mean excitatory and inhibitory currents. 
Discrepancies in firing rates between network and reduced models were 
investigated, and mechanisms leading to these discrepancies were identified. 
Chief among them is correlations in spiking, or partial synchronization, 
working in concert with ``nonlinearities"
in the time evolution of membrane potentials. Additionally, simple random walk models and their first passage times were shown to reproduce well fluctuations 
in neuronal membrane potentials and interspike times.}

\bigskip
\section*{Introduction}

Models in neuroscience come in an extraordinarily wide range, from very simple,
seeking to describe complex neural behavior using a few coarse-grained variables, 
to extremely complex, as in Connectome type projects 
that seek to provide a complete map of all neuronal connections -- with a myriad of possibilities in
between. The modeling of individual neurons alone can vary from a single number
that describes its firing rate, to an integrate-and-fire equation, or 
a Hodgkin-Huxley model, or one that treats individual ionic channels and the 
biochemical reactions that
are triggered with each synapse. Needless to say, models that incorporate
more neuroanatomical and neurophysiological details are more realistic, but 
realistic modeling is not without cost: with greater complexity comes more 
unknown parameters corresponding to quantities that cannot be measured 
in the laboratory. 

A question that we believe has not received adequate attention is: how do models
of different levels of complexity compare? One does not expect a reduced model 
to provide the same kind
of detailed information as a large-scale network of spiking neurons, but does 
it provide useful information, rough but in the ballpark? 
If not, what causes the discrepancies? What are the mechanisms
reduced models lack that cause them to produce inaccurate results?

Clearly, these questions cannot be asked in the abstract; the answer will
depend not only on the models but on the type of questions. Until the theory
is more advanced, one will have to severely limit one's scope. In this paper, we
report on findings from a study that compares some specific kinds of models
of local neuronal populations. Our ``detailed" models are stochastic models
of interacting neurons with integrate-and-fire type dynamics; they can be
thought of as modeling the dynamical interactions that take place in local 
circuitries in the mammalian 
cortex; and our reduced models are either described by simple ODEs of 
Wilson-Cowan type or by random walks in the space of membrane
potentials. 

It soon became clear that even after specifying what types of models to use, the
question is still too broad to be tackled: networks of
interacting neurons defined by simple equations can have diverse behaviors
depending on coupling and other parameters. Some of these behaviors are direct consequence of
parameter choices; others are emergent, that is, they occur as a result of
the dynamical interaction among neurons. There is no such thing
as ``typical" network behavior.

To select representative network models in a meaningful way, it was
necessary to first identify the issues of interest. In our case, the most salient
features of our ``detailed" models that are not present in reduced models are
correlations in subthreshold and spiking activity that emerge from neuronal
interactions. Rate models do not treat individual neurons as entities, 
and do not therefore have the capacity to describe such correlations. 
This prompted us to look at not a single network model but  a collection of 
networks with different amounts of correlations,
using as a starting point a network that mimics the realistic model 
of the visual cortex in \cite{chariker2016orientation}.

There is an earlier paper that had a similar goal as ours, namely
\cite{grabska2014well}, which compared 
networks of integrate-and-fire neurons with corresponding 
mean field models. They observed, as we did, that the degree of synchrony
mattered.
An important difference between the present paper and \cite{grabska2014well} is
that we sought to identify the underlying mechanisms 
that led to discrepancies between network and reduced models. To our knowledge this is
the first time such an analysis has been carried out.

\medskip
The organization of this paper is as follows: In Section 1, we introduce
a class of stochastic models of interacting neurons. These will be our ``detailed"
models. A mathematical treatment of these models
is given in Section 2; this section can be skipped if the reader so chooses.
In Section 3, we produce some network models with different degrees of
synchrony, to be used for comparison with reduced models.
In Section 4, we considered reduced models defined by simple ODEs,
and in Section 5, we modeled fluctuations in membrane potential as random-walks.

\bigskip
\section{A stochastic model of interacting neurons}

In this section, we introduce a stochastic model of interacting neurons 
representing a local population in the cerebral cortex. Though not intended
to depict any specific part of the brain, this model has some of the features 
of realistic cortical models. Importantly, 
the dynamics are driven by neuronal interactions, with all the attendant
correlated spiking behaviors that emerge as a result of these 
interactions. We have elected to use a stochastic model because with 
the aid of ergodicity, firing rates are represented simply and convergence is fast.
The model presented here will be used in the rest of this paper to evaluate 
the performance of reduced models that are much simpler.

\medskip
\subsection{Model description} 
We consider a population of neurons connected by local circuitry, such as
neurons in an orientation domain of one layer of the primary visual cortex.
We assume that this population contains $N_{E}$ excitatory (E) neurons
and $N_{I}$ inhibitory (I) neurons, which are sparsely and homogeneously
connected. The following
assumptions are made in order to formulate a simple
Markov process that describes the spiking activity of this population. 

\bigskip
\begin{itemize}
\item[(1)] We assume for simplicity that the membrane potential of a neuron 
 takes only finitely many discrete values.
\item[(2)] Each neuron receives synaptic input from an external source in the form 
of Poisson kicks; these kicks are independent from neuron to neuron.
 \item[(3)] When the membrane potential of a neuron reaches a certain threshold,
 the neuron spikes, after which it goes to a refractory state and remains there for an exponentially distributed random time.
\item[(4)] Every time an E (respectively I) neuron in the population spikes:  

(a) a random set of postsynaptic neurons is chosen;
  
(b) the membrane potential of each chosen postsynaptic neuron goes up 

\qquad (respectively down) after an exponentially distributed random time.
 \end{itemize}

\bigskip
More precisely, we assume that in our population there are $N_E$ 
excitatory neurons,
labeled $1,2,\cdots, N_E$, and $N_I$ inhibitory neurons, labeled $N_E+1, N_E+2,
\cdots, N_E + N_I$. The membrane potential of neuron $i$, denoted $V_i$, takes
values in 
$$\Gamma := \{-M_r, -M_r+1, \cdots, -1, 0,1,2, \cdots, M\} \cup \{\mathcal R\}\ .
$$
Here $M, M_r \in \mathbb Z^+$; $M$ represents the threshold for spiking, 
$-M_r$ the inhibitory reversal potential, and $\mathcal R$ the refractory state. 
When $V_i$ reaches $M$, the neuron is said to {\it spike}, and $V_i$ is
instantaneously set to $\mathcal R$, where it remains for a duration given by
an exponential random variable with mean $\tau_{\mathcal R}>0$. When $V_i$ 
leaves $\mathcal R$, it goes to $0$.

We describe first the effects of the ``external drive", external in the sense that
this input comes from outside of the population in question; for example it can be thought of as coming from a region of cortex closer to sensory input. 
This input is delivered in the form of impulsive kicks, arriving at random
(Poissonian) times, the Poisson processes being independent from neuron to 
neuron. For simplicity, we assume 
there are two parameters, $\lambda^E, \lambda^I>0$, representing 
the rates of the Poisson kicks to the E and I neurons in the 
population. These rates are low in background; they 
increase with the strength of the stimulation.
When a kick arrives and $V_i \ne \mathcal R$, $V_i$ jumps up by $1$, until it reaches $M$, at which time the neuron spikes. Kicks received
by neuron $i$ when $V_i=\mathcal  R$ have no effect.

Each neuron also receives synaptic input from within the population.
We assume that an excitatory kick received by a neuron ``takes
effect" (the meaning of which will be made precise momentarily) at  
a random time after its arrival. This delay is given by an exponential
random variable with mean $\tau^E$; it is independent from spike to spike and
from neuron to neuron. Similarly,
an inhibitory kick received takes effect after a random time with mean $\tau^I$.
We let $H^E_i$ denote the number of E-kicks received by neuron $i$ that
has not yet taken effect, and let $H^I_i$ denote the corresponding number of
I-kicks. That is to say, every time an E-kick is received by neuron $i$, $H^E_i$
goes up by $1$; every time an E-kick received by neurons $i$ takes effect, 
$H^E_i$ goes down by $1$,
and so on. The state of neuron $i$ at any one moment in time
 is then described by the triplet 
$(V_i, H^E_i, H^I_i)$. We will refer to $H^E_i$ and $H^I_i$ as the numbers of kicks
``in play"; these two numbers may be viewed 
as stand-ins for the E and I-conductances of neuron $i$.

We now explain what it means for an E or I-kick to take effect.
Each E or I-kick received by neuron $i$ carries with it an (independent)
exponential clock as discussed above. When this clock rings,
what happens depends on 
whether or not $V_i = \mathcal R$. If $V_i = \mathcal R$,
then $V_i$ is unchanged. If $V_i \ne \mathcal R$, then
$V_i$ is modified instantaneously according to the numbers $S_{Q,Q'}, \ 
Q, Q' \in \{ E, I\}$, where $S_{Q,Q'}$ denotes the synaptic coupling when
a neuron of type $Q'$ synapses on a neuron of type $Q$. In the case of an
I-kick, this modification also depends on $V_i$.

Here is how $V_i$ is modified in the case of an E-kick, i.e., when $Q'=E$: 
Assume first
that the numbers $S_{Q,Q'}$ are nonnegative integers. When an E-neuron
spikes and it synapses on neuron $i$, $V_i$ jumps up by $S_{EE}$
if $i$ is an E-neuron, by $S_{IE}$ if $i$ is an I-neuron; and if the jump takes 
$V_i$ to an integer $ \ge M$, it simply goes to $\mathcal R$. 
For non-integer values of $S_{Q,Q'}$, let $p =
\left \lfloor S_{Q,Q'} \right \rfloor$ be the greatest integer
less than or equal to $S_{Q,Q'}$. Then $S_{Q,Q'} = p + u$ 
where $u$ be a Bernoulli random
 variable taking values in $\{0,1\}$ with $\mathbb{P}[ u = 1] = S_{Q,Q'} - p$ 
   independent of all other random variables in the model. 

When I-spikes take effect, the rule is analogous to that for E-spikes,
with the following exception: $V_i$ jumps down instead of up by an amount proportional to $V_i+M_r$, where $-M_r$ is the reversal potential
for I-currents. The numbers $S_{Q,Q'}$ are assumed
to be positive, and for definiteness, let us declare $S_{Q,I}$ to be the size 
of the jump at $V_i=M$, so that in general, the size of the
jump is $$
S_{Q, I}(V_{i}) := (V_i+M_r)/(M+M_r)*S_{Q,I} \,.
$$ 

We remark that we have incorporated into the numbers $S_{Q,Q'}$
the changes in {\it current} in the postsynaptic neuron. 
We have assumed that E-currents are independent of the membrane potential
of the postsynaptic neuron, which is not unreasonable since the reversal potential
for excitatory current is quite large ($>4M$ in our setup).
Changes in I-current are more sensitive to membrane potential, and that is
reflected in the formula above.
  
It remains to stipulate the ``connectivity" of the network, i.e., the set of 
neurons postsynaptic to each neuron. We assume for
simplicity that connectivity in our model is random and time-dependent, so that 
every time a neuron spikes, a random set of postsynaptic neurons
is chosen anew (independently of history). More precisely, for $Q,Q' \in \{E,I\}$, 
we let
 $P_{Q,Q'} \in [0,1]$ be the probability that a neuron of type $Q$ is postsynaptic 
when a neuron of type $Q'$ spikes, and the set of postsynaptic neurons is
determined by a coin flip with these probabilities following each spike. 
We do not pretend this assumption is realistic; in the real brain
connectivities between neurons are fixed and far from random. But unlike
longer range projections, which tend to target specific regions or even neurons, 
exact connectivities within local populations are not known
to be important. This is a rationale behind our assumption of random
postsynaptic neurons. Another is that this assumption simplifies the analysis considerably. In particular, it makes the behaviors of all neurons in the 
E-population, respectively the I-population, statistically indistinguishable.

This completes our description of the model.

\medskip
\subsection{Parameters used in numerical results}

There is an analytical and a numerical part to our results.
Our rigorous results apply to all parameter choices that satisfy the 
hypotheses in the theorems or propositions. We give a sense here of
the parameters we use in simulations: We generally take
$N_E$ to range from $300$ to $1000$, and $N_I= \frac13 N_E$, 
as is typically the case in local populations in the real cortex.
We set $M = 100$, $M_r=66$, the ratio of $M_r$ to
$M$ reflecting biologically realistic ranges of membrane potentials. We fix
$P_{EE} = 0.15$, $P_{IE} = P_{EI} = 0.5$ and $P_{II}=0.4$, these
numbers chosen to resemble the usual connectivities in networks such as those 
in the visual cortex. There is less experimental guidance for 
the synaptic couplings $S_{Q,Q'}$; we take them to be $2-6$, out of the $100$
units between reset and threshold (cf $S_{EE} =5$ means it takes 20 
consecutive E-kicks to drive a neuron from $V_i=0$ to $V_i=M$).
We set $\tau_{\mathcal R}=2-3$ (ms), 
consistent with usual estimates for refractory periods, and 
set $\tau^E$ and $\tau^I$ to be a few ms, with $\tau^E < \tau^I$, 
as AMPA is known to act faster than GABA and both act within milliseconds.
We will, on occasion, deliberately choose parameters that are a little unbiological
 to make a point. 
Finally, the Poisson rates of the external drive, $\lambda^E, \lambda^I$ will 
be varied as we study the model's responses to drives of various strengths.

\bigskip
Readers who wish to bypass the technical mathematics pertaining to 
the class of models described above can proceed without difficulty 
to Section 3.

\bigskip
\section{Theoretical results and proofs} 

Some basic results for the model presented in Sect. 1.1 are stated in
Sect. 2.1, and their proofs are given in Sect. 2.3, after a brief review of
probabilistic preliminaries. 

\medskip 
\subsection{Statement of results}

The model described above is that of a Markov jump process $\Phi_{t}$ 
on a countable state space
$$
\mathbf{X} = ( \Gamma \times \mathbb{Z}_{+} \times \mathbb{Z}_{+}
)^{N_{E} + N_{I}}\, ,
$$
as the state of neuron $i$ is given by the triplet $(V_i, H^E_i, H^I_i)$ where
$V_i \in \Gamma$ and $H^E_i, H^I_i \in \mathbb{Z}_{+}:= \{0,1,2,\dots\}$.
We assume the paths of $\Phi_{t}$ are c\`adl\`ag. 
The transition probabilities of $\Phi_{t}$ are denoted by
$P^{t}(\mathbf{x}, \mathbf{y})$, i.e.,
$$
  P^{t}(\mathbf{x},\mathbf{y}) = \mathbb{P}[ \Phi_{t} = \mathbf{y}
  \,|\, \Phi_{0} = \mathbf{x}] \,. 
$$ 
The left operator of $P^{t}$ acting on a probability distribution
$\mu$ is
$$
  \mu P^{t}  (\mathbf{x}) = \sum_{\mathbf{y} \in \mathbf{X}} \mu(
  \mathbf{y})P^{t}( \mathbf{y}, \mathbf{x}) \,.
$$
The right operator of $P^{t}$ acting on an observable $\xi: \mathbf{X}
\rightarrow \mathbb{R}$ is
$$
  P^{t} \xi ( \mathbf{x}) = \sum_{\mathbf{y} \in \mathbf{X}} P^{t}(
  \mathbf{x}, \mathbf{y})\xi( \mathbf{y}) \,.
$$

Our first result pertains to
the existence and uniqueness, hence ergodicity, of invariant measure
for the Markov chain $\Phi_{t}$. Notice that as $H^E_i$ and $H^I_i$ can
take arbitrarily large values, the state space for $\Phi_{t}$ is noncompact, 
and such Markov chains need not possess invariant probabilities
in general.

For $U: \mathbf{X} \mapsto (0, \infty)$, we define the $U$-weighted
total variation norm of a signed measure $\mu$ on $\mathcal{B}(\mathbf{X})$,
the Borel $\sigma$-algebra of $\mathbf{X}$, to be
$$
  \| \mu \|_{U} = \sum_{\mathbf{x} \in \mathbf{X}} U(\mathbf{x})|\mu( \mathbf{x})| \,,
$$
and let 
$$
  L_{U}(\mathbf{X}) = \{ \mu \mbox{ on } \mathcal{B}(\mathbf{X}) \,|\,
  \| \mu \|_{U} < \infty \} \,.
$$
 To state the main result, we need the
following definitions. For each state $\mathbf{x} \in \mathbf{X}$,
we let 
$$
  H^E(\mathbf{x}) = \sum_{i = 1}^{N_{E}+N_I} H^{E}_{i}
  \qquad \mbox{and} \qquad
  H^I(\mathbf{x}) = \sum_{i = 1}^{N_{E}+N_I} H^{I}_{i}
  $$
be the total number of E-kicks and I-kicks in play.
  
\begin{thm}
\label{invariant} The Markov chain
$\Phi_{t}$ admits a unique invariant probability measure $\pi \in
L_{U}( \mathbf{X})$ where 
$$
U(\mathbf{x}) = H^E( \mathbf{x}) + H^I( \mathbf{x}) + 1\ .
$$
This stationary measure is ergodic with exponential
convergence to equilibrium, equivalently exponential decay of time correlations.
More precisely, there exist constants
$C_{1}, C_{2}>0$ and
$r \in (0, 1)$, such that 
\begin{itemize}
  \item[(a)] for any initial distribution $\mu \in L_{U}(\mathbf{X})$,
$$
  \| \mu P^{t} - \pi \|_{U} \leq C_{1} r^{t} \| \mu - \pi \|_{U}\,;
$$
\item[(b)] 
for any observable $\xi$ with $\|\xi \|_{U} < \infty$,
$$
  \| P^{t} \xi - \pi( \xi) \|_{U} \leq C_{2} r^{t} \| \xi - \pi( \xi )
  \|_{U} \,,
$$ 
where
$$
  \pi( \xi) = \sum_{ \mathbf{x} \in \mathbf{X}} \pi(\mathbf{x}) \xi(
  \mathbf{x}) \,.
$$
\end{itemize}
\end{thm}

\bigskip
For $T >0$, we let $N_i(T)$ denote the number
of times neuron $i$ spikes on the time interval $[0,T]$. Equivalently, 
$N_i(T)$ is the random variable that counts the number of visits of $V_i$
to state $M$ (or to state $\mathcal R$). 
Then the {\it mean firing rate} of neuron
$i$ is defined to be
$$
\mbox{fr}(i) := \lim_{T \to \infty} \frac1T  \mathbb E_\pi [N_i(T)] =  \mathbb E_\pi [N_i(1)]
$$
where $\mathbb E_\pi$ is the expectation with respect to the invariant
probability $\pi$ given by Theorem 2.1. As defined, $\mbox{fr}(i)$ is a number
in $[0, \infty]$, and the second equality follows
from the invariance of $\pi$. Since all E-neurons have the same
firing rates, and the same is true for all I-neurons, we denote the mean E- 
and I-firing rates of $\Phi_t$ by $\bar F_E$ and $\bar F_I$ respectively. 

In the next corollary, we assume $S_{Q,Q'}$
are integers, and leave the formulation of results for noninteger values of 
$S_{Q,Q'}$ to the reader. 
For each state $\mathbf{x} \in \mathbf{X}$, we define
\begin{eqnarray*}
F_E^{\rm tot}(\mathbf{x}) & = & \sum_{i=1}^{N_E} 
\left(\frac{1}{\tau^E} H^E_i \cdot \mathbf{1}_{ \{ V_{i} \geq M - S_{EE}\} }(\mathbf{x}) + 
\lambda^E \cdot \mathbf{1}_{\{V_{i} = M - 1 \}}(\mathbf{x}) \right)\ ,\\
F_I^{\rm tot}(\mathbf{x}) & = &  \sum_{i=N_E+1}^{N_E+N_I} 
\left(\frac{1}{\tau^E} H^E_i \cdot \mathbf{1}_{ \{ V_{i} \geq M - S_{IE}\} }(\mathbf{x}) + 
\lambda^I \cdot \mathbf{1}_{\{V_{i} = M - 1 \}}(\mathbf{x}) \right)\ .
\end{eqnarray*}

\begin{cor}
\label{cor1} The firing rates $\bar F_E,
\bar F_I < \infty$ and satisfy
$$
\bar F_E = \frac{1}{N_E} \sum_{\mathbf{x} \in \mathbf{X}} 
 F^{\rm tot}_E( \mathbf{x}) \pi(\mathbf{x})\ 
\quad \mbox{and} \quad 
\bar F_I = \frac{1}{N_I} \sum_{\mathbf{x} \in \mathbf{X}} 
F^{\rm tot}_I( \mathbf{x}) \pi(\mathbf{x}) \,.
$$
\end{cor}

\medskip
\subsection{Probabilistic Preliminaries} We review the following general results
on geometric ergodicity.
Let $\Psi_{n}$ be a Markov chain on a countable state space $(X, \mathcal{B})$ 
with transition kernels
$\mathcal{P}(x, \cdot)$, and let $W: X \rightarrow [1, \infty)$. 
Consider the following conditions:

\begin{itemize}
  \item[(a)] There exist constants $K \geq 0$ and $\gamma \in (0, 1)$
    such that
$$
  (\mathcal{P}W)(x) \leq \gamma W(x) + K
$$ 
for all $x \in X$.
\item[(b)] There exists a constant $\alpha \in (0, 1)$ and a
  probability measure $\nu$ so that
$$
  \inf_{x\in C} \mathcal{P}(x, \cdot) \geq \alpha \nu(\cdot) \,,
$$
with $C = \{x \in X \, | \, W(x) \leq R \}$ for some $R > 2K(1 -
\gamma)$, where $K $ and $\gamma$ are from (a).
\end{itemize}

\medskip

The following was first proved in \cite{meyn2009markov}. The version
we use is proved in \cite{hairer2011yet}.

\begin{thm}
\label{hairer}
Assume (a) and (b). Then $\Psi_{n}$ admits a unique invariant measure
$\pi \in L_{W}(X)$. In addition, there exist constants $C, C' > 0$ and
$r \in (0, 1)$ such that (ii) for all $\mu, \nu \in L_{W}(X)$, 
$$
  \| \mu \mathcal{P}^{n} - \nu \mathcal{P}^{n} \|_{W} \leq C r^{n} \|
  \mu - \nu \|_{W} \,,
$$ 
and (i) for all $\xi$ with $\|\xi\|_{W} < \infty$,
$$
  \| \mathcal{P}^{n} \xi - \pi( \xi) \|_{W} \leq C' r^{n} \| \xi -
  \pi(\xi) \|_{W} \,.
$$
\end{thm}

\medskip
\subsection{Proof of Theorem 2.1 and Corollary 2.2.}
For a step size $h > 0$, we define the time-$h$ sample chain as
$\Phi^{h}_{n} = \Phi_{nh}$, and  drop the superscript $h$ when it leads
to no confusion. We first show for this discrete-time chain that 
$U(\mathbf{x}) = H^E(\mathbf{x}) + H^I(\mathbf{x})+1$ is a natural 
Lyapunov function that satisfies conditions (a) and (b) in the previous subsection. 

\begin{lem}
\label{cond1}
For $h > 0$ sufficiently small, there exist constants $K
> 0$ and $\gamma \in (0, 1)$, such that
$$
  P^{h}U \leq \gamma U + K \,.
$$
\end{lem}

Intuitively, this is true because on a short time interval $(0,h)$, $U$ decreases
at a rate proportional to $H^E+H^I$ as kicks received prior to time $0$
take effect, while it can increase at most by a fixed constant related 
to $N_E +N_I$ due to the presence of the refractory period.

\begin{proof}
We have
$$
  P^{h}U( \mathbf{x}) = \mathbb{E}_{\mathbf{x}}[ U( \Phi_{h})] \leq U(
  \mathbf{x}) - \mathbb{E}_{\mathbf{x}}[ N_{out}] + \mathbb{E}_{\mathbf{x}}[N_{in}] \,,
$$
where $N_{out}$ is the number of kicks from $H^E+H^I$ that takes
 effect on $(0,h]$, and 
$N_{in}$ is the number of new spikes produced during the time period $(0, h]$.

To estimate $N_{out}$, recall that the clocks associated with each of the $H^E+H^I$
kicks are independent, with each E-kick taking effect on $(0,h]$
with probability $(1-e^{- h/ \tau^E})$ and each I-kick taking effect
on $(0,h]$ with probability $(1-e^{-h/\tau^I})$. This gives
$$
 \mathbb{E}_{\mathbf{x}}[ N_{out}] \ge (H^E + H^I) ( 1 - e^{-h/\max\{\tau^E, \tau^I\}} ) \geq  \frac{1}{2 \max\{\tau^E, \tau^I\}}  \ h \ U( \mathbf{x})
$$
for $h$ sufficiently small. 

To estimate $N_{in}$, consider neuron $i$, and let $f_i$ be the number
of spikes generated by neuron $i$ during the time period $(0, h]$. 
Since after each spike neuron $i$ spends an exponential time with mean
$\tau_{\mathcal R}$
in state $\mathcal R$, we have
$$
  \mathbb{E}_{\mathbf{x}}[f_{i}] \leq 1 + \mathbb{E}[ \mbox{Poisson
    distribution with parameter } h/\tau_{\mathcal R}] = 1 + h/\tau_{\mathcal R} \,.
$$
Hence
$$
  \mathbb{E}_{\mathbf{x}}[N_{in}] \leq (N_{E} + N_{I})(1 + h/\tau_{\mathcal R}) \,.
$$
The proof is completed by letting 
$$
\gamma = 1 - h/(2 \max\{\tau^E, \tau^I\}) \qquad \mbox{and} \qquad
K =  (N_{E} + N_{I})(1 + h/\tau_{\mathcal R})\ .
$$
\end{proof}

For $b \in \mathbb R$, let
$$
C_{b} = \{  \mathbf{x} \in \mathbf{X} | H^E( \mathbf{x}) + H^I( \mathbf{x})
\leq b \} \,. 
$$

\begin{lem}
\label{cond2}
Let $\mathbf{x}_{0}$ be the state that $H^E = H^I = 0$ and  $V_i = \mathcal R$
for all $i$. Then for any $h > 0$, there exists a constant $c$ depending on $b$ 
and $h$ such that for all $\mathbf{x} \in C_{b}$,
$$
  P^{h}(\mathbf{x}, \mathbf{x}_{0}) > c \,.
$$
\end{lem}
\begin{proof}
It is sufficient to construct, for each $\mathbf{x} \in C_{b}$, a sample path 
that goes from $\mathbf{x}$ to
$\mathbf{x}_{0}$ with a uniform lower bound on its probability. 
Consider the following sequence of events.

\begin{itemize}
  \item[(i)] A sequence of Poisson kicks increases each $V_i$
   to the threshold value $M$, hence puts $V_i=\mathcal R$, by time $t=h/2$; once in
$\mathcal R$, $V_i$ remains there through time $t = h$.  
No kick in play takes effect on $[0,h/2]$.
\item[(ii)] All kicks in play at time $h/2$ take effect on $(h/2, h]$, but that
has no effect as all $V_i$  are in $\mathcal R$.
\end{itemize}

To prove that the events above occur with a positive probability bounded
from below, observe that in the scenario described, the number of kicks in play
never exceeds $b+N_E +N_I$, hence only a finite number of conditions are
imposed. 
\end{proof}

Lemmas \ref{cond1} and \ref{cond2} together imply Theorem \ref{invariant}.

\begin{proof}[Proof of Theorem \ref{invariant}] 
Choose step size $h$ as in Lemma \ref{cond1}. It follows from Lemma
\ref{cond1} and \ref{cond2} that the assumptions in Theorem
\ref{hairer} are satisfied. Therefore, the discrete-time chain $\Phi^{h}$
admits a unique invariant probability measure $\pi_{h}$ in $L_{U}(
\mathbf{X})$.

\medskip

We will show that $\pi_{h}$ is invariant under $\Phi_t$ for any $t > 0$. This is
because $\Phi_{t}$ satisfies the ``continuity at zero'' condition, meaning for
any probability measure $\mu$ on $\mathbf{X}$,
$$
  \lim_{t \rightarrow 0}\| \mu P^{t} - \mu \|_{TV} = 0 \,.
$$
To see this, let $\epsilon > 0$ be an arbitrary small number. Since
$\mu$ is finite, there exists a finite set $A\subset \mathbf{X}$ such
that $\mu(A) > 1 - \epsilon/4$. By the definition of $A$, clock rates for initial values in $A$ are uniformly
bounded. Therefore, one can find a sufficiently small $\delta > 0$,
such that $\mathbb{P}[ \mbox{ no clock rings on } [0, \delta) ] \geq 1
- \epsilon/4$. For any set $U \subset \mathbf{X}$, we have
\begin{eqnarray*}
(\mu P^{\delta}) (U)  & =  & \sum_{\mathbf{x} \in \mathbf{X}}
P^{\delta}( \mathbf{x} , U) \mu( \mathbf{x} )\\
&=& \sum_{\mathbf{x}\in A \cap U} P^{\delta}(\mathbf{x}, U) \mu( \mathbf{x} ) + \sum_{\mathbf{x} \in A -
  U} P^{\delta}(\mathbf{x}, U) \mu( \mathbf{x} ) + \sum_{\mathbf{x} \in A_{c}} P^{\delta}(\mathbf{x}, U) \mu( \mathbf{x}) \\
&=& \mu(A \cap U) - a_{1} + a_{2} + a_{3} \,,
\end{eqnarray*}
where
\begin{eqnarray*}
a_{1} & =  & \sum_{\mathbf{x} \in A \cap U} (1 - P^{\delta}(\mathbf{x}, U) )\mu(\mathbf{x}) \leq \frac{\epsilon}{4} \mu( A \cap U) \leq \frac{\epsilon}{4}\\
a_{2} &=& \sum_{\mathbf{x} \in A \setminus U} P^{\delta}(\mathbf{x}, U) \mu( \mathbf{x}) \leq \frac{\epsilon}{4} \mu( A \setminus U) \leq \frac{\epsilon}{4}\\
a_{3} &=& \sum_{\mathbf{x} \in A^{c}} \frac{\epsilon}{4} \mu( A \cap U) \leq
\frac{\epsilon}{4} \leq \mu( A^{c}) \leq \frac{\epsilon}{4} \,.
\end{eqnarray*}
In addition we have $\mu(U) - \mu(A \cap U) \leq \mu(A^{c}) <
\frac{\epsilon}{4}$. Hence
$$
  | (\mu P^{\delta})(U) - \mu(U) | < \epsilon
$$
for any $U \subset \mathbb{R}^{N}_{+}$. By the definition of the total
variation norm, we have
$$
  \| \mu P^{\delta} - \mu \|_{TV} \leq \epsilon \,.
$$
This implies the ``continuity at zero'' condition.

Notice that $\pi_{h}$ is invariant for any $\Psi^{hj/k}_{n}$, where
$j, k \in \mathbb{Z}^{+}$ (Theorem 10.4.5 of \cite{meyn2009markov}). Then without loss of generality, assume $t/h \notin
\mathbb{Q}$. By the density of orbits in irrational rotations, there
exists sequences $a_{n}$, $b_{n} \in \mathbb{Z}^{+}$ such that
$$
  d_{n} := t - \frac{a_{n}}{b_{n}}h \rightarrow 0
$$
from right. Then 
$$
  \pi_{h} P^{t} = \pi_{h}P^{\frac{a_{n}}{b_{n}} h}
  P^{d_{n}} \,.
$$
Therefore,
$$
\| \pi_{h}P^{t} - \pi_{h} \|_{TV} \leq   \lim_{n \rightarrow \infty}
\| \pi_{h}P^{d_{n}} - \pi_{h} \|_{TV}  = 0
$$
by the ``continuity at zero'' condition. Hence
$\pi_{h}$ is invariant with respect to $P^{t}$. 

\medskip

It remains to prove the exponential convergence for any $t > 0$. By
Lemma \ref{cond1}, there exists $B$ such that $P^{t}V \leq
B V$ for all $t < h$. Let $n$ be the largest integer that is smaller than $t/h$ and let $d = t
- hn$. Then we have

\begin{eqnarray*}
  \| \mu P^{t} - \nu P^{t}\|_{V} & = & \| (\mu P^{d}) P^{nh} - (\nu
  P^{d}) P^{hn}\|_{V}\\
&  = & C r^{n} \cdot \| \mu P^{d} - \nu P^{d} \|_{V}\  \leq \ BC r^{n}
  \| \mu - \nu \|_{V} \,.
\end{eqnarray*}

Similarly, 
\begin{eqnarray*}
  \| P^{t} \xi - \pi(\xi)\|_{V} & = & \| P^{nh} (P^{d \xi})- P^{hn} (P^{d}
  \pi(\xi))\|_{V}\\
 & = & C r^{n} \cdot \| P^{d} \xi - P^{d} \pi(\xi) \|_{V} \ \leq \ BC r^{n}
  \| \xi - \pi(\xi) \|_{V} \,.
\end{eqnarray*}
This completes the proof.
\end{proof}

\medskip
\begin{proof}[Proof of Corollary 2.2] As in the proof of Lemma 2.4, we have
that for every $\mathbf{x}$,
$$
  \mathbb{E}_{\mathbf{x}}[N_i(1)] \leq 1 + \mathbb{E}[ \mbox{ Poisson
    distribution with parameter } 1/\tau_{\mathcal R}] = 1 + 1/\tau_{\mathcal R} \,.
$$ 
Thus $\mathrm{fr}(i) = \mathbb{E}_{\pi}[N_i(1)] < \infty$.

Now for an infinitesimally small $h>0$, the probability that a spike is fired
by some E-neuron on the time interval $[0,h]$ starting from $\mathbf{x}$
is $F_E^{\rm tot}(\mathbf{x}) \cdot h$. Thus starting from $\mathbf{x} = \Phi_0$,
$$
\bar F_E = \frac{1}{N_E} \cdot \lim_{T\to \infty} \frac1T \int_0^T 
F_E^{\rm tot}(\Phi_t) dt
$$
which by the Ergodic Theorem is equal to the asserted quantity. 
The same argument applies to $\bar F_I$.
\end{proof}

\medskip
\section{Three populations with different degrees of synchrony}

One of the reasons why a rigorous analysis of the dynamics of populations 
of interacting neurons is challenging  is that these interactions produce 
correlations in temporal dynamics, and coordinated dynamical events
 in turn influence
how the neurons interact. Correlations in neuronal dynamics clearly
exist in the real brain: LFP oscillations in the gamma band, which are found 
in many parts of cortex, are likely the result of self-organized, coordinated subthreshold activity.  A related phenomenon is clustering or partial 
synchronization of spikes. An extreme version of it involving full-population
spikes is known as PING \cite{borgers2003synchronization}. 
Milder, more subtle, forms 
of partial synchronization that appear to be more consistent with 
experimental observations \cite{henrie2005lfp} 
in cortex were later identified in the modeling work of
\cite{rangan2013emergent} and studied in \cite{rangan2013dynamics, chariker2015emergent}; similar phenomena were reproduced in the realistic models in \cite{chariker2016orientation} and \cite{chariker2017rhythm}. 

In other developments in neuroscience, many reduced models 
(e.g. \cite{wilson1972excitatory, wilson1973mathematical, knight1996dynamical,
amit1997model, omurtag2000dynamics, haskell2001population}) 
have been proposed, giving rise to simpler ways to deduce the firing rates 
of local populations. 
As discussed in the Introduction, one of the aims of this paper is to investigate
the accuracy of these predictions, and the degree to which they may (or may not) be affected by emergent correlations in network activity. In the rest of this paper, 
we will make believe that the stochastic model presented in Sect. 1.1 is 
``real", and its firing rates will be compared to those given
by some of the simplest reduced models.

While our intentions are clear, a question arises immediately with regard to which
parameters that define the network models in Sect. 1.1 to use. Numerical
exploration tells us quickly that there is no such thing as a typical network 
model, and that different parameter choices can lead to a wide range of
dynamical behaviors. Obviously we cannot exhaust all possibilities in
a high dimensional parameter space. Since our main concern is the effect 
of correlated spiking, we will introduce in this section
three examples of network models exhibiting different degrees of correlated 
spiking or partial synchrony, and use them for illustration in the next section.

\medskip
\subsection{Three example networks}

We introduce here three models of the type described in Sect. 1.1, with
 identical parameters except for $\tau^E$ and $\tau^I$, the expected time
between the occurrence of a spike and when it takes effect. As we will see,
different choices of these values  will lead to different degrees of synchrony.  

We first give the parameters common to all three models: $N_E$ and $N_I$,
the numbers of E and I neurons in the population, are 300 and 100 respectively. 
The connectivities $P_{QQ'}$ are as in Sect. 2.1, namely
$P_{EE} = 0.15$, $P_{IE} = P_{EI} = 0.5$ and $P_{II}=0.4$. The synaptic
weights $S_{QQ'}$ are as follows: $S_{EE}=5$, $S_{IE}=2$, and $S_{EI}=
S_{II}=4.91$ (recall that this corresponds to coupling weights when $V_i$,
the membrane potential of the postsynaptic neuron, is at threshold, i.e. at 
$V_i=100$). The expected time to stay in the refractory state, $\mathcal R$,
is $2.5$ ms, and the external drive rates to E and I neurons will be taken
to be equal, i.e., $\lambda^E=\lambda^I$, and a range of values of the drive
will be considered.

The parameters above, for the most part, are similar to those used in the
realistic models of visual cortex \cite{chariker2016orientation}. From 
\cite{chariker2017rhythm}, we learned also that
varying the rise and decay times of E and I conductances, especially the relation
between the two, is a very effective way to change the degree of synchrony of
a local population. We now use this technique to produce the following
three examples:

\bigskip \noindent
(1) The ``{\bf homogeneous}" network, abbreviated as ``{\bf Hom}" in the figures:
$$
\tau^{EE} = 4 , \quad \tau^{IE} = 1.2 , \quad \tau^I = 4.5 \qquad \mbox{(in ms)}
$$
(2) The ``{\bf regular}" network, abbreviated as ``{\bf Reg}" in the figures:
$$
\tau^{EE} = 2.0 , \quad \tau^{IE} = 1.2 , \quad \tau^I = 4.5 \qquad \mbox{(in ms)}
$$
(3) The ``{\bf synchronized}" network, abbreviated as ``{\bf Sync}" in the figures:
$$
\tau^{EE} = 1.3 , \quad \tau^{IE} = 0.95 , \quad \tau^I = 4.5 \qquad \mbox{(in ms)}
$$

Instead of $\tau^E$, we have used here $\tau^{EE}$ and $\tau^{IE}$ to denote
the expected times between the occurrence of an E-spike and when it takes 
effect in E, respectively I, neurons. These numbers are roughly consistent with
biological values: $\tau^I>\tau^{EE}, \tau^{IE}$ is consistent with the fact that
GABA acts more slowly than AMPA, and $\tau^{EE} > \tau^{IE}$ is consistent with 
the fact that E-spikes can synapse on dendrites of E-neurons, taking 
a bit longer for its effect to reach the soma, while they synapse directly on
the soma of I-cells. That aside, there is nothing special about these choices,
other than that they produce the distinct degrees of synchrony that we would like
to have.

Figure \ref{fig1} shows the E- and I-firing rates of the three networks 
above in response to a range of drives of magnitude 
$\lambda = \lambda^E=\lambda^I$ spikes/sec. 
Both firing rates increase monotonically as a function of drive.
We think of $\lambda \sim 1000$ spikes/sec as low drive, or spontaneous activity, 
and $\lambda \ge 6000$ spikes/sec as strong drive.

\begin{figure}[htbp]
\centerline{\includegraphics[width = 0.7\textwidth]{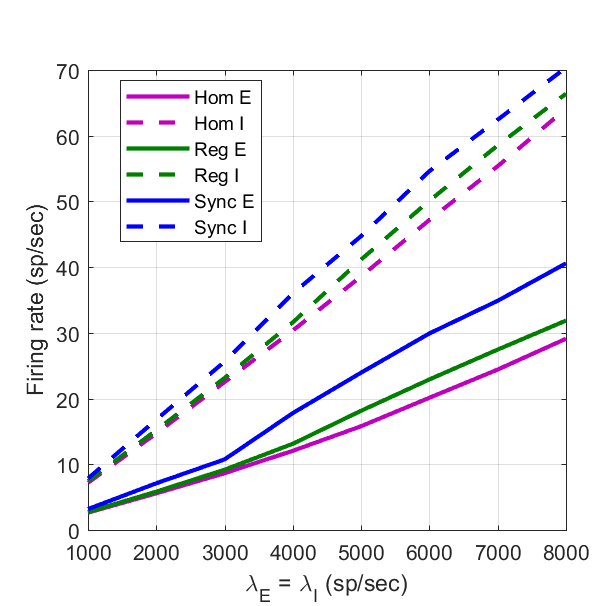}}
\caption{ {\bf Firing rates of three example networks in response to
increasing drive.} In the $x$-axis, $\lambda=\lambda^E=\lambda^I$
is external drive. The graphs labeled ``Hom", ``Reg" and ``Sync" give
the firing rates of the corresponding networks.}
\label{fig1}
\end{figure}

\medskip
\subsection{Statistics of the ``Hom", ``Reg" and ``Sync" networks.} Here we present more detailed
information on the three networks defined in the last subsection, focusing on
their responses to relatively strong drive, at $\lambda=7000$ spikes/sec.
Fig 2 shows, for each network, spike rasters, summed spike plots, and 
correlation diagrams. 

The raster-plots are self-explanatory. Clearly visible in the rasters of
the Reg and Synch networks are coordinated spiking that 
emulate gamma band oscillations (at 30-90 Hz) in the real cortex \cite{henrie2005lfp}. 
These spiking events are entirely emergent, or 
self-organized, in the sense that there is nothing built into the network 
architecture or dynamics that lead directly to these spiking events. 
Comparing the frequency of these events with mean E-firing rate
(given above the rasters),
one sees  that most E-neurons do not participate in all spiking events.

The summed spike plots give the fractions of the E-population spiking
in $5$-ms time bins. Though they show the same behaviors as the rasters, 
we have included
these plots because rasters can be deceiving when 
used to depict the spiking activity of hundreds of neurons: what 
appear to be population spikes may in fact involve fewer neurons than
the rasters suggest. For the Sync network, one sees from the summed
spike plots that most spiking events do not involve the entire population, 
even though the rasters may give an impression to the contrary. 
As for the Reg network, Fig 2 shows that the larger spiking events 
usually involve no more than $30-40$\% of
the population. Nor do identical fractions of neurons spike in each 5 ms bin
in the Hom network: there is some amount of
synchronization that is entirely emergent, natural and hard to avoid.

The correlation diagrams describe not correlations between pairs of 
neurons but how the spiking of individual neurons are correlated to that
of the rest of the population. We describe precisely what is plotted in, for example,
the second histogram from the left, labeled ``Conditioned on E at $t=0$" with
an ``I" in the box. 
Here we run the network for 10-20 seconds. Each time an E-neuron spikes,
we record all the I-spikes fired within $15$ ms of its occurrence, both before
and after, computing the fraction of the I-population spiking in each
$1$-ms time bin on this time interval. This is then averaged over all E-spikes
that occur during the simulation.
The other three plots are interpreted analogously. A comparison of these plots 
for the three networks confirms the increasing amounts of correlated spiking, or 
partial synchrony, that are clearly visible in the rasters as we
go from the Hom to the Sync network.

\begin{figure}[htbp]
\centerline{\includegraphics[trim=0in .2in 0in .2in, clip,width = \linewidth]{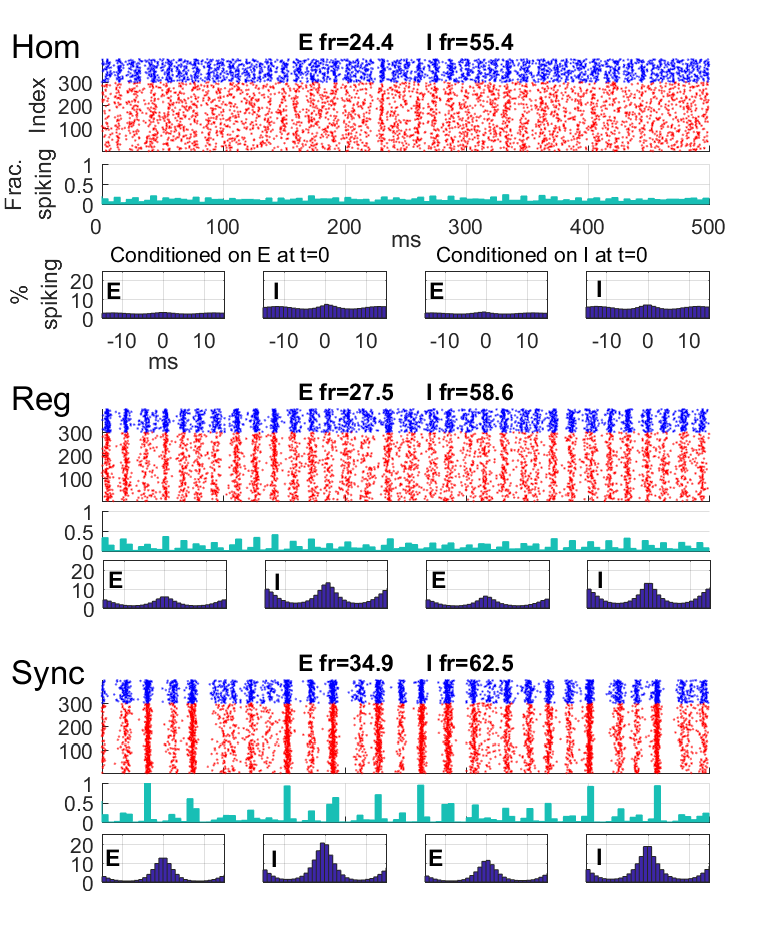}}
\caption{{\bf Statistics of the Hom, Reg and Sync networks:} 
All statistics are collected in response to a strong drive of
$\lambda^E=\lambda^I=7000$ spikes/sec. For each network we show 
in the top panel rasters (E-neurons in red, I-neurons in blue) over 
a 1/2 sec time interval; mean firing rates are shown above
the rasters. In the middle panel are 
corresponding summed spike plots for E-neurons, showing 
the percentage of the E-population
spiking in each 5 ms window. Below the summed spike plots are correlation diagrams: a histogram labeled ``$X$ conditioned on $Y$ at $t=0$", 
$X,Y = E, I$, shows the percentage of the $X$-population spiking on 1 ms 
windows on the time interval $t \in [-15, 15]$ ms conditioned on a $Y$-spike 
occurring at time $t=0$. Labels for the Reg and Sync networks, which are
omitted, are to be read as identical
to those for the Hom network. }
\end{figure}

\bigskip
\noindent
{\bf Analysis.} We have presented three example networks 
defined by essentially the same parameters yet exhibiting remarkably dissimilar
spiking patterns, from very homogeneous to strongly synchronized. 
The only differences in network parameters are $\tau^E$ and 
$\tau^I$, which describe how long after 
one neuron synapses on another before the effect of the spike is fully felt.
Even here, the differences are subtle: the homogeneous and regular networks
differ only in $\tau^{EE}$ and by only $2$ ms, while all three $\tau$-parameters
differ by $<1$ ms between the regular and synchronized networks. 

Two points here are of note. First, when under drive, the most salient kind 
of correlations among neurons in the model are semi-regular bursts of 
elevated spiking occurring with frequencies in the gamma band (not to
suggest that these are the only correlations). Second, our simulations confirm that
small changes in $\tau^E$ and $\tau^I$, intended to represent
how AMPA and GABA affect conductance  properties in the postsynaptic
neuron in the real brain, have a strong impact on the amount of correlated
spiking or degree of synchronization in the local
population. 

The mechanism behind gamma band oscillations has been much studied.
An extreme form of it involving full population spikes, called PING, was first 
described in \cite{borgers2003synchronization}. Milder and more realistic forms  producing spectral power
densities much closer to data were studied in
\cite{rangan2013dynamics, chariker2015emergent} and \cite{chariker2017rhythm}. We refer
the reader to these papers for a more detailed discussion. Very briefly, these
rhythms occur as a result of recurrent excitation and the fact that the time course 
for GABA is generally a few ms slower than that of AMPA, allowing 
some fraction of the E- and
I-population to spike before a sufficient amount of GABA is released
to curb the spiking activity.

Finally, to be clear, we do not claim that the examples above are representative of
all network models. If anything, they illustrate that neuronal interactions 
can produce a wide range of dynamical characteristics, and that these 
characteristics can depend on model parameters in subtle ways.
But with partial synchronization being one of the most salient features
of driven neuronal dynamics, these three examples allow us close-up looks 
into how reduced models perform when used to predict the dynamics of
networks with different degrees of synchronization.

\bigskip

\section{Firing rates: comparison of reduced and network models}

Up until now, we have focused on models defined by populations
of interacting neurons. We now turn to the use of reduced models to 
estimate their firing rates. Three very simple ODEs describing the evolution of membrane potential are proposed in Sect. 4.1.
No novelty is claimed here with regard to these reduced models; 
many similar ideas for deducing firing rates
by balancing one quantity or another have been 
proposed in the literature (see e.g. \cite{wilson1972excitatory, wilson1973mathematical,
knight1996dynamical, amit1997dynamics, amit1997model, 
van1998chaotic, omurtag2000dynamics, gerstner2000population,
 brunel2000dynamics,  haskell2001population,  cai2004effective, cai2006kinetic,
rangan2006maximum}). 
 The reduced models we have selected for consideration
were chosen for their simplicity, and the fact that they allow a direct comparison 
with the network models studied in Sections 1--3. 
Such a comparison is carried out in Sect. 4.2, followed by an analysis
of the discrepancies.

\medskip
\subsection{Three reduced models and their firing rates} The models below 
will be referred to by their names in italics in later discussion.

\medskip \noindent
(1)  {\it Linear model}. In this first reduced model we regard the membrane potential $v$ of each neuron as drifting upward at constant speed, i.e., 
\begin{equation} \label{voltage}
\frac{dv}{dt} = F^+ - F^- \ , \qquad \mbox{for} \ v \in [0,1]\ .
\end{equation}
Upon reaching $1$, $v$ is instantaneously reset to $0$, and the climb starts
immediately (with no refractory period). Here $F^+$ and
$F^-$ are the forces that drive $v$ upward, respectively downward. 
They are connected to the quantities that describe the network models in Sect. 1.1
as follows: Let 
\begin{eqnarray}
C_{EE} = N_{E}P_{EE}S_{EE}\ , & \qquad &  C_{IE} = N_{E}P_{IE}S_{IE}\ , \\\nonumber
C_{EI} = N_{I}P_{EI}\hat{S}_{EI} \ , & \qquad &  C_{II} = N_{I}P_{II}\hat{S}_{II}\ ,
\end{eqnarray}
where $\hat{S}_{EI}$ and $\hat{S}_{II}$ are to be taken to be the value of $S_{EI}(v)$ and $S_{II}(v)$
at $v=M/2=50$ in the model in Sect. 1.1. Then for E-neurons, 
$$
F^+ = \frac{1}{M} (f_E * C_{EE} + \lambda^E) \qquad \mbox{and} \qquad
F^- =  \frac{1}{M} \ f_I * C_{EI}\ ,
$$
and for I-neurons,
$$
F^+ = \frac{1}{M} (f_E * C_{IE} + \lambda^I)  \qquad \mbox{and} \qquad
F^- = \frac{1}{M} \ f_I * C_{II}\ .
$$
Here $f_E$ and $f_I$ are to be thought of as mean E- and I-firing rates of the population. 
 
The mean excitatory and inhibitory 
firing rates $f^{(1)}_E$ and $f^{(1)}_I$ of this reduced model are 
defined to be the values of $f_E$ and $f_I$ 
that satisfy the self-consistency condition that when these values are
plugged into the equations above, they produce the same firing rates 
(the number of times  per sec $v$ in (\ref{voltage}) reaches $1$). 
They can be computed explicitly as follows:
 
\begin{lem} The values $(f^{(1)}_E, f^{(1)}_I)$ are uniquely defined and
are given by the formulas below, provided the quantities on the right side
are $\ge 0$:
\begin{eqnarray} \label{meanrate}
  f^{(1)}_{E} &=& \frac{\lambda^{E}(M + C_{II}) - \lambda^{I}C_{EI} }{(M - C_{EE})(M +
  C_{II}) + (C_{EI}C_{IE})} \\\nonumber
 f^{(1)}_{I} &=& \frac{\lambda^{I}(M - C_{EE}) + \lambda^{E}C_{IE} }{(M - C_{EE})(M +
  C_{II})+ (C_{EI}C_{IE})} \,,
\end{eqnarray}
\end{lem}

\medskip
With $\lambda^E=\lambda^I$ as we have done in Section 3, it is easy to see that $f^{(1)}_E$ and $f^{(1)}_I$ increase linearly as functions of drive.  
 
\bigskip \noindent
(2) {\it Linear model with refractory}. This model is similar to the previous one, 
except for the presence of a (fixed) refractory period. That is to say,
here 
$$
\frac{dv}{dt} = F^+ - F^- \ , \qquad \mbox{for} \ v \in [0,1]\ ,
$$
except that every time $v$ reaches $1$ and is reset to $0$, it remains there
for a time interval of length $\tau_{\mathcal R}$ before resuming its linear climb. 
See Fig 3A (second from left).

The mean E- and I-firing rates of this model, $(f^{(2)}_E, f^{(2)}_I)$, are
then given by the pair $(f_E, f_I)$ satisfying the quadratic equations
\begin{eqnarray}
\label{quad}
 M * f_E & = & (1-\tau_{\mathcal R}f_{E})(f_E * C_{EE} + \lambda^E  - f_I *  C_{EI})\\
 M * f_I & = & (1-\tau_{\mathcal R}f_{E})(f_E * C_{IE} + \lambda^E  - f_I *  C_{II})\ .
\end{eqnarray}
Theoretically, these equations can be solved analytically. From the
first equation of \eqref{quad}, it is easy to see that
$$
  f_{I} = \frac{(1 - \tau_{\mathcal{R}} f_{E})(C_{EE} f_{E} + \lambda^{E}) - M f_{E}
  }{C_{EI}(1 - \tau_{\mathcal{R}} f_{E})}  \,.
$$
Putting $f_{I}$ into the second equation of \eqref{quad} and multiplying both
sides by $(1 - \tau_{\mathcal{R}} f_{E})^{2}$, we obtain a quartic equation 
for $f_{E}$ of the form 
$$
  A_{0}(\tau_{\mathcal{R}}) + A_{1}(\tau_{\mathcal{R}})f_{E} + A_{2}(\tau_{\mathcal{R}})f_{E}^{2} +
  A_{3}(\tau_{\mathcal{R}})f_{E}^{3} + A_{4}(\tau_{\mathcal{R}})f_{E}^{4} = 0 \,,
$$
where $A_{0} \sim A_{4}$ are polynomials in terms of
  $\tau_{\mathcal{R}}$. In particular, when $\tau_{\mathcal{R}} = 0$, this quartic equation reduces to the
linear equation
\begin{equation}
\label{linearE}
  [(M - C_{EE})(M + C_{II}) + (C_{EI}C_{IE})] f_{E} = \lambda^{E}(M +
  C_{II}) - \lambda^{I}C_{EI} \,,
\end{equation}
which produces $f_{E}^{(1)}$ in \eqref{meanrate}. It is
well known that quartic equations have a root formula, which is
unfortunately too complicated to be practical, but it
gives the existence of solution to the
quadratic system. In addition, for sufficiently small $\tau_{\mathcal{R}}$, the
quartic equation is a small perturbation of equation
\eqref{linearE}. By the intermediate value theorem, it is easy to show
that the quartic equation must admit a root that is close to
$f_{E}^{(1)}$. We leave this elementary proof to the reader.

\bigskip \noindent
(3) The $v$-{\it dependent} model. Here $v$ satisfies the same equation as before,
except that $S_{EI}$ and $S_{II}$ depend on the distance of $v$ to the reversal
potential. To separate the effects of refractory and $v$-dependence of synaptic
weights, let us assume for definiteness that there is no refractory period, that is
to say, all is as in the linear model except for the following: 
$$
S_{EI}(v) = \frac{Mv+ M_{r}}{M+ M_{r}}*S_{EI} \qquad \mbox{and} \qquad
S_{II}(v) =  \frac{Mv+ M_{r}}{M+ M_{r}}*S_{II}\ .
$$

For a given pair $(f_{E}, f_{I})$, this gives us two first order linear ODEs
$$
\frac{\mathrm{d}v_{E}}{\mathrm{d}t}  =   A_{E} - B_{E} v_{E}
\qquad \mbox{and} \qquad
\frac{\mathrm{d}v_{I}}{\mathrm{d}t} = A_{I} - B_{I} v_{I} \,,
$$
where $A_E, A_I, B_E$ and $B_I$ are easily computed from network
parameters. 
We let $t_{E}$ and $t_{I}$ be the times $v_E$ and $v_I$ first reaches $1$.
Then the desired spike rates $f_E$ and $f_I$  should  satisfy
$f_E = t_E^{-1}$ and $f_I=t_I^{-1}$. That is, the firing rates $f^{(3)}_E$ and
$f^{(3)}_I$ of this ODE model is the pair $(f_E, f_I)$ that solves
the two nonlinear equations
\begin{equation}
  \label{vdp}
1 = \frac{A_{E}}{B_{E}}(1 - e^{- B_{E}f_{E}^{-1}}) \quad, \quad 1 =
  \frac{A_{I}}{B_{I}}(1 - e^{- B_{I}f_{I}^{-1}}) \,.
\end{equation}
These equations can be solved numerically.

\bigskip
Needless to say, one can also consider the combined effects of (2) and (3),
to obtain a $v$-dependent model with refractory. In this case, equation
\eqref{vdp} becomes
$$
  1 = \frac{A_{E}}{B_{E}}(1 - e^{- B_{E}(f_{E}^{-1} - \tau_{\mathcal{R}})} ) \quad, \quad 1 =
  \frac{A_{I}}{B_{I}}(1 - e^{- B_{I}(f_{I}^{-1} - \tau_{\mathcal{R}})}) \,.
$$

\bigskip
Firing rates for the first two reduced models are shown in Fig. 3B (left) using
the parameters of the network models studied in Section 3. An immediate
observation is that the model with refractory has higher firing rates, which
may be somewhat counter-intuitive as the delay during refractory should, 
on the face of it, lead to lower firing rates. We have omitted
 the firing rates for the $v$-dependent model because they are ridiculously
low (close to $0$) and numerically unstable, and that requires an explanation
as well. (Please ignore the plot with open circles for now.) 

\medskip \noindent
{\bf Analysis.} The following is a heuristic explanation for $f^{(2)}_E > f^{(1)}_E$:
As is usually the case, I-firing rate is significantly higher than E-firing rate
in the models considered. With refractory, every time a neuron spikes, 
it ``misses" some amount of the incoming drive, the net value of which is positive. Since the fraction of drive ``missed" is proportional to the firing rate of a neuron
in these models, the I-neuron ``misses" a larger fraction of its drive than the
E-neuron. This may cause the system to become more excited than in the case 
of no refractory. (In the argument above we have taken into 
account first order effects only, ignoring the secondary effect that 
higher E-firing will boost I-firing.)

With regard to the v-dependent model, our analysis shows that the root
$(f^{(3)}_E, f^{(3)}_I)$ is very sensitive with respect to small
change of constants $B_E$ and $B_I$. Small errors in $B_E$ or $B_I$
caused by the inhomogeneous arrival of spikes are dramatically
amplified by the v-dependent model. 
As a result, the computed values are usually too low and not sufficiently reliable 
to be useful.

That begs the question then: why are firing rates in the network models so much
higher than in the $v$-dependent model, and so robust? 
We believe stochastic fluctuations
is the answer, and will study that in the next section.

\begin{figure}[htbp]
\centerline{\includegraphics[width = \linewidth]{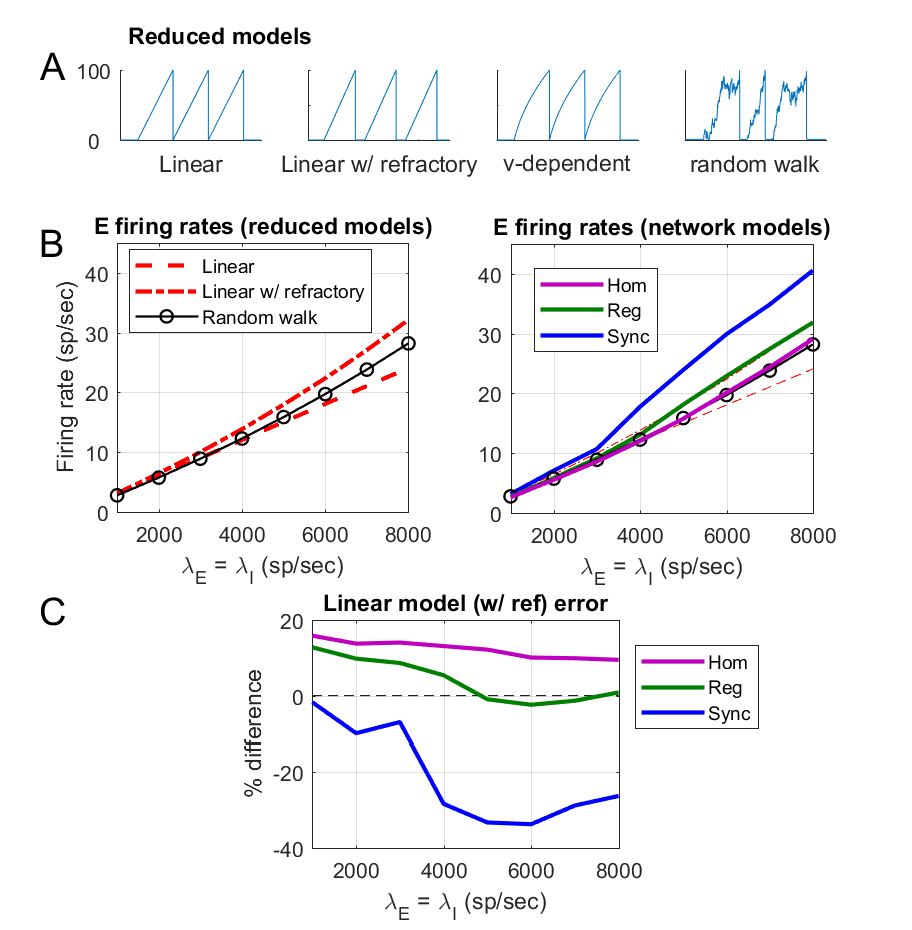}}
\label{fig3}
\caption{{\bf Comparison of firing rates.} {\bf A.}
Trajectories of the membrane potential $v$ as functions of time, for 4 reduced
models. From left to right: the linear model, the linear model with refractory, 
the $v$-dependent model, and the random-walk model considered in Sect. 5.
{\bf B.} The left panel shows graphs of E-firing rates as functions of drive
of the two linear reduced models, with and without refractory, and of the 
random walk model discussed in Section 5 (black with open circles). 
Firing rates of the $v$-dependent model
are omitted as they cannot be computed reliably. The right panel
shows firing rates of the network models (from Fig 1) superimposed on
the graphs from the left panel. {\bf C.} Percentage error if one uses the linear
model with refractory to predict firing rates of network models.
For example, $-20\%$ means the reduced model predicts a firing
rate $20\%$ lower than that of the network model. }
\end{figure}

\medskip
\subsection{Comparison of firing rates with network models.} 
We now compare the firing rates of the network models in Sect. 3 and the reduced models in Sect. 4.1. The right panel of Fig 3B shows the 
firing rates of the two linear reduced models (with and without refractory)
copied from the panel on the left and superimposed on the firing rates 
of the Hom, Reg and Sync models copied from Fig 1. We see immediately
that the linear reduced model
underestimates the firing rates of all three networks for moderate and strong
drives; and the linear model
with refractory, which has higher firing rates as explained earlier, underestimates
the firing rates of the Sync model and overestimates that of the Hom model.
Fig. 3C gives the percentage errors if the linear model with refractory was used to 
predict the firing rates of the network models. It confirms what is shown in
Fig 3B (right).

For definiteness, we now focus on a single reduced model, namely 
the linear model with refractory, and refer to it simply as 
``the reduced model" in the rest of this section.
There are likely many reasons why network firing rates are not in total agreement
those of this reduced model. We will focus on two of them: correlated
spiking in the form of partial synchronization as depicted in Figure 2, and the 
$V$-dependence of I-currents. The reason correlated spiking, or synchronization,
may be relevant is that in this reduced model, the arrival of synaptic input to 
a neuron is assumed to be homogeneous in time. Indeed such an assumption is implicit (or explicit) in most reduced models, 
even though it is in direct contradiction to correlations in spiking or partial 
synchronization, phenomena that are well known to occur in the real brain.

One of the effects of correlated spiking is that a disproportionately large
fraction of synaptic input may be missed during refractory.
Some statistics pertinent to our investigation are shown in Fig 4. 
The bar graphs in Fig 4A show the percentages
of E and I-spikes missed during refractory. Here we have distinguished between
spikes from interactions among neurons within the population and from external
drive. As external drive is assumed to be constant in time, 
one may equate the percentage
of spikes from external drive missed with the percentage of time a neuron
spends in refractory. As expected, the percentages of
E- and I-spikes missed in the Hom network are reasonably close to those 
predicted by the reduced model. In the Sync network, 
the percentages of synaptic input
missed are considerably higher than the percentage of time spent in refractory, 
consistent with the fact that spike times in this network are strongly correlated; 
see the correlation diagrams in Fig 2. 
That a smaller percentage of I-spikes are missed than E-spikes
is likely due to $\tau^I$ being large relative to $\tau_{\mathcal R}$, so 
more I-spikes arrive after the neuron leaves refractory. 
Fig 4B shows the empirical mean values of $V$ in the three network models,
which are well above the mean values of $V$ when I-spikes take effect. 

\begin{figure}[htbp]
\centerline{\includegraphics[width = \linewidth]{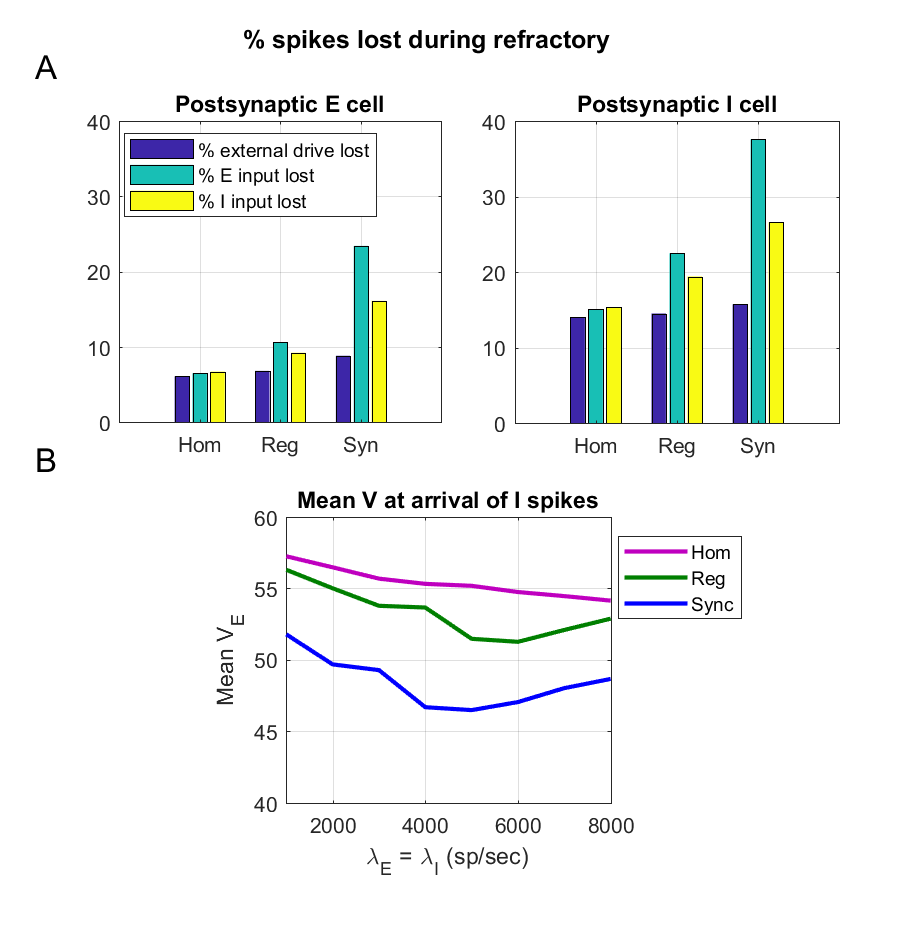}}
\label{fig4} 
\caption{{\bf Challenging the homogeneity of drive assumption in reduced models.} {\bf A.} Percentages of external drive, E and I-spikes 
missed during refractory, for E and I neurons, for the Hom, Reg and Sync networks,
for postsynaptic E-cells (left) and postsynaptic I-cells (right).
The percentage of external drive missed can be taken to be $\%$ time spent 
in refractory; E and I input here refer to synaptic input from within the population. 
{\bf B.} Mean $V$-values of E neurons (dashed) and mean $V$-values
(solid) when I-spikes take effect, 
as functions of drive, for the three networks. Corresponding graphs for I-neurons
are qualitatively similar.}
\end{figure}

\bigskip \noindent
{\bf Analysis.} We now attempt to explain the deviations of network firing
rates from those predicted by the linear model with refractory.

\medskip \noindent
(a) {\it Discrepancy caused by lack of $V$-dependence of I-currents.} 
Our reduced model used $S_{EI}$ and $S_{II}$ values that correspond to 
network values at $V=M/2 =50$. This choice is based on the assumption
that $V$ marches at constant speed from reset to threshold, and $I$-spikes
arrive in a time-homogeneous way, neither of which is true and the situation
is complicated: That I-spikes are stronger for larger $V$ should cause mean $V$ 
to be $>50$, but strong synchronization may cause more 
I-spikes to arrive when $V$ is lower. Indeed according to Fig 4B, at drive=7000 sp/s,
the mean $V$-value when I-spikes take effect are $\sim 54.5, 52$ and $48$ respectively
for the Hom, Reg and Sync networks. This means
using $S_{EI}$ and $S_{II}$ values at $V=M/2$
underestimates the mean values of these parameters for 
the Hom and Reg networks, and overestimates it for the Sync network. Underestimating $S_{EI}$ means that the network is
in fact more suppressed than this reduced model suggests. To summarize:
based on this one property alone, we would expect the reduced model to have
a higher firing rate than the Hom and Reg networks (with a smaller error
for the Reg network) and to have a lower firing rate than the Sync network.

\medskip \noindent
(b) {\it The effects of partial synchronization working in concert with refractory.} 
Because the arrival of synaptic input is not necessarily homogeneous in time,
the fraction of E and I-spikes ``missed" during refractory can be nontrivially
altered by partial synchronization. There is no easy way to predict the net
effect of this phenomenon, however, because it involves both E- and I-inputs missed by both E and I neurons, leading to a not-so-simple cancellation problem.
 
Consider first E-neurons. Suppose an {\it additional} fraction 
$\varepsilon_{EE}$ of E-input, and an additional fraction $\varepsilon_{EI}$ 
of I-input, to E-neurons are lost during refractory -- ``additional" in the sense
that it is above and beyond what is assumed to be lost during refractory
under the homogeneity of input assumption. Then compared to the reduced
model, there is a net gain in (positive) current in the amount of
$$
\Delta F_E \ = \ \varepsilon_{EI} f_I C_{EI} - \varepsilon_{EE} f_E C_{EE} 
= 172 \ \varepsilon_{EI} f_I - 225 \ \varepsilon_{EE} f_E \ .
$$
(This number can be positive or negative.) For an I-neuron, net gain relative to the reduced model is
$$
\Delta F_I \ = \ \varepsilon_{II} f_I C_{II} - \varepsilon_{IE} f_E C_{IE} 
= 137 \ \varepsilon_{II} f_I - 300 \ \varepsilon_{IE} f_E \ .
$$

For the Sync model, we see from Fig 4A that
$$
\varepsilon_{EI} \approx    8 , \quad 
\varepsilon_{EE} \approx 14.5 , \quad
\varepsilon_{II} \approx  11 , \quad
\varepsilon_{IE} \approx 22\ .
$$
With $\frac32 f_E< f_I < 3 f_E$ (Fig 1), it is easy to see that $\Delta F_I$ is
significantly more negative than $\Delta F_E$. That is to say, synchronization
causes I-neurons to lose more (positive) input current than E-neurons, 
so the system should be more excited and E-firing rate should be higher than predicted by the
reduced model. 

An analogous computation gives the same conclusion for the Reg network, 
but the difference between $\Delta F_E$ and $\Delta F_I$ is smaller. 
The $\varepsilon_{QQ'}$ values for the Hom network are too small to be significant.

\medskip
Combining (a) and (b), we expect that the linear model with refractory
will give E-firing rates that are higher than the Hom network (counting only
the error from (a)), and lower than the Sync network (errors from both (a) and (b)). 
As for the Reg network, the two errors from (a) and (b) have opposite signs; 
one cannot say what it will be on balance except that it is likely to be smaller 
than the other two. This is consistent with the results in Fig 3B (right).

\bigskip \noindent
{\it A general remark.} We have found  that $\Delta F_I$ is generally larger 
in magnitude than $\Delta F_E$ when input
currents are changed, due simply to the fact that I-neurons have higher firing rates,
we believe. These changes depend on the composition of the current that is altered,
however, and in the situation above that depends on the relative speeds at which
E and I spikes take effect, i.e., on $\tau^E$ and $\tau^I$. A complete analysis
of that is beyond the scope of this paper. We believe a separate mathematical model about the interaction
  of E and I spikes during a spiking event is necessary to address
  this issue. Thus while we have often
seen that synchronization leads to higher firing rates, we do not know if this is
always the case, or the conditions under which this is
true.

\bigskip
\section{Modeling membrane potentials as random walks}  

In this section, we consider a different kind of reduced model, namely one
in which membrane potentials of E and I neurons are modeled as (biased) 
random walks with reset at threshold. This model is in part motivated
by the fact that reduced models defined by ODEs cannot reproduce the
statistics of events observed in populations of interacting neurons, 
and that stochastic fluctuations -- or population activity that give rise 
to behaviors that resemble stochastic 
fluctuations -- seem to play a role in neuronal dynamics. 

\subsection{A random walk model and its firing rate.}

Here we model the membrane potentials of an E and an I-neuron by a
continuous time Markov jump process $(X^E_t, X^I_t)$ where $X^E_t$ and 
$X^I_t$ are independent and each takes values in the state space
$\{-M_{r}, \cdots, M-1, \mathcal R\}$. For definiteness we will consider a model
that incorporates both refractory periods and the $V$-dependence of I-currents.

Given a pair $(f_E, f_I)$ which represents the firing rates of excitatory 
and inhibitory neurons from the local population, we assume
that $X^E_t$ is driven by three independent Poisson processes that correspond
to (i) external drive, (ii) excitatory and (iii) inhibitory synaptic inputs from the population. The Poisson process corresponding to (i) delivers kicks
at rate $\lambda^E$. The ones corresponding to (ii) and (iii)
have rate $N_E P_{EE} f_E$ and $N_I P_{EI} f_I$ respectively. 
Upon receiving a kick from the external drive,
$X^t_E$ moves up by $1$. Upon receiving a kick from (ii), $X^E_t$ jumps up by $S_{EE}$ slots, and upon receiving a kick from (iii), it jumps down by 
$S_{EI}(X^E_t)$ slots. The interpretation of non-integer
numbers of slots and the $X^E_t$-dependence of $S_{EI}$ are as in
Sect. 1.1. Also as before, when $X^E_t$ reaches $M$, it goes to $\mathcal R$, 
where it remains for an exponential time of mean $\tau_{\mathcal R}$.
The process $X^I_t$ is defined analogously.

It is well known that an irreducible Markov jump process on a finite
state space admits a unique stationary distribution. 
Given $(f_E, f_I)$, let $\nu_Q$ denote the stationary distribution of $X^Q_t$
for $Q=E,I$. Clearly, $\nu_{Q}$ is a computable distribution satisfying
\begin{equation}
  \label{linear}
\left \{ 
\begin{array}{l}
  \mathbf{A}_{Q} \nu_{Q} = \mathbf{0}\\
\mathbf{1}^{T} \nu_{Q} = 1
\end{array}
\right .
\end{equation}
where $\mathbf{A}_{Q}$ is the generator matrix of process $X^{Q}_{t}$, and
$\mathbf{1}$ is a vector in $\mathbb{R}^{M + M_{r} + 1}$ all of whose entries are equal to $1$. The {\it firing rate} of $X^{Q}_{t}$, $Q\in \{E, I\}$, can be defined as
$$
  \tilde{f}_{Q} = \lim_{T \rightarrow \infty}\frac{1}{T}\#\{ t \in (0, T) \,|\,
 X^{Q}_{t-} \neq \mathcal R, X^{Q}_{t} = \mathcal R \} \,.
$$
It is easy to see that
$$
  \tilde{f}_{Q} = N_{E}P_{QE}f_{E}\sum_{i = M - S_{QE}}^{M-1}\nu_{Q}(i) +
  \lambda^{Q}\nu_{Q}(M-1) \,.
$$

Of interest to us is $(f_E,f_I)$ satisfying the consistency condition
$(f_E,f_I)=(\tilde f_E,\tilde f_I)$. We prove the existence of a solution
to this consistency equation.

\begin{thm}
\label{rwrate}
There exist $f_{E}, f_{I}>0$ such that when
$(X^{E}_{t},X^{I}_{t})$ is driven by these firing rates, they produce
mean firing rates $\tilde f_{E}$ and $\tilde f_{I}$ such that
$$
  \tilde{f}_{E} = f_{E}, \quad \tilde{f}_{I} = f_{I} \,.
$$
\end{thm}

\begin{proof}
Let
$$
  \phi_{1}(f_{E}, f_{I}) = \tilde{f}_{E} - f_{E} \quad \mbox{and}
  \quad \phi_{2}(f_{E}, f_{I}) = \tilde{f}_{I} - f_{I} \,.
$$
Observe that for any $f_I \ge 0$,
$$
\phi_{1}(0, f_{I}) > 0 \qquad \mbox{and} \qquad
\phi_{1}(f_{E}, f_{I}) < 0 \quad \mbox{when} \quad f_{E} > \tau_{\mathcal R}^{-1}\ .
$$
The first inequality is true because independently
of how high a rate the inhibitory clock rings, there exists $T_0>0$ and 
$\varepsilon >0$ (depending on $f_I$) such that starting from anywhere 
in $\Gamma$, external drive alone will, with probability $\ge \varepsilon$,
cause $X^E_t$ to spike within $T_0$ units of time, rendering
$\tilde f_E>0$. The second inequality is true because
each time $X^E_t$ spikes, it has to spend time in refractory, 
so $\tilde f_E \le \tau_{\mathcal R}^{-1}$. 
Similarly, observe that for any $f_E \ge 0$, 
$$
\phi_{2}(f_{E}, 0) > 0 \qquad \mbox{and} \qquad
\phi_{2}(f_{E}, f_{I}) < 0 \quad \mbox{when} \quad f_{I} > \tau_{\mathcal R}^{-1}\ .
$$
By the Poincare-Miranda Theorem (a version of
intermediate value theorem in dimensions greater than one), 
we have the existence of a solution
$$
  \phi_{1}(f_{E}, f_{I}) = 0, \quad \phi_{2}(f_{E}, f_{I}) = 0 \,.
$$
Moreover, from the boundary conditions above, we have that
$f_E, f_I >0$.
\end{proof}

\begin{figure}[htbp]
\centerline{\includegraphics[width = \linewidth]{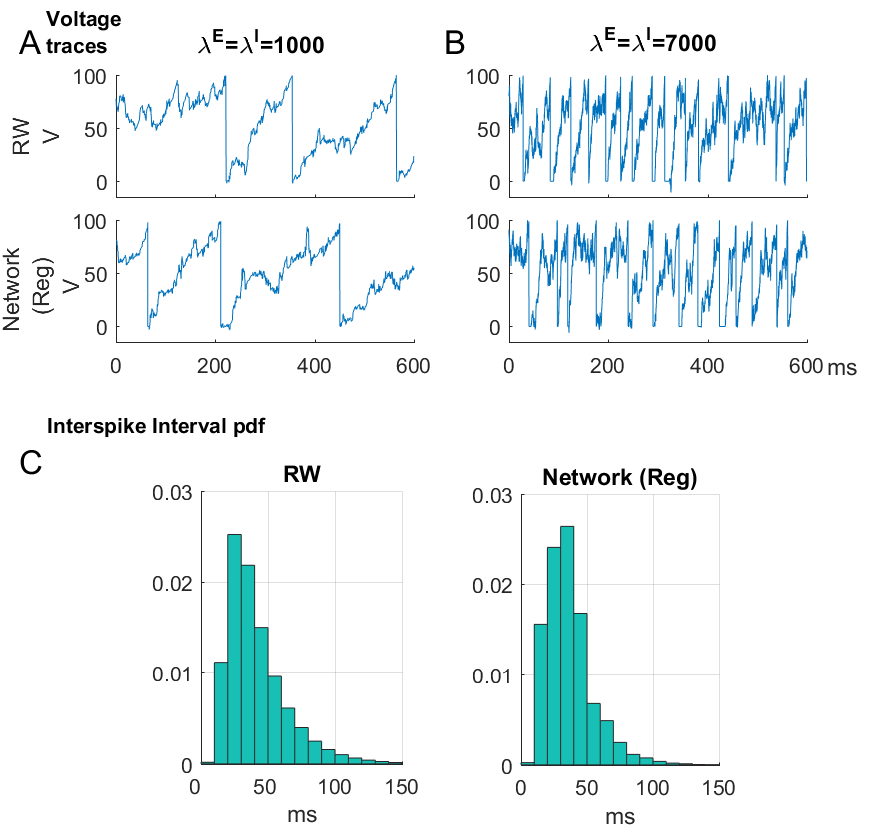}}
\label{fig5} 
\caption{{\bf Trajectory of single neuron in network model vs. random walk.}
  A-B. Sample paths of $X^E_t$ are shown above traces of membrane potential
of a randomly chosen neuron in the Reg network. A: background. B: strongly
driven.  C. Empirical distributions of interspike intervals for $X^E_t$ (left) and
for a neuron from the Reg network. In the definition of $X^E_t$,
the same parameters as in the network models are used in the
Poisson processes, and the firing rates used are $(f^{rw}_E, f^{rw}_I)$
from Theorem 5.1.}
\end{figure}

Let $f_{E}^{rw}$ and $f_{I}^{rw}$ denote the mean firing rates obtained 
from Theorem \ref{rwrate}. They were found to be unique in our numerical 
simulations, and very close to the empirical firing rates of the Hom network. 
See Fig 3B, the graph in black with open circles. That this graph is close
to the one for the Hom network and somewhat below those of the Reg and Sync network models is consistent with our analysis in Sect. 4.2: Here we have
corrected the error due to $V$-dependence but not the one due to the
combined action of synchronization and refractory. As explained in Sect. 4.2, 
such action causes E-firing rates of the Sync and Reg networks to be higher
than that predicted under the assumption that the arrival of synaptic input is homogeneous in time. 

Sample paths of $X^E_t$ in background ($\lambda^E=\lambda^I=1000$) and
when strongly driven ($\lambda^E=\lambda^I=7000$) are shown in Fig 5A,B (top).
We have included for comparison voltage traces of an E-neuron from the 
Reg model at corresponding drives directly below. These traces are very similar and
not easily discernible by eye.

\medskip
\subsection{Interspike intervals.}

As noted earlier, the random walk (rw) model above has the capability of
producing statistics that may emulate those in network models, something
the reduced ODE models studied earlier cannot do. In this subsection,
we focus on the distribution of interspike times, i.e., the times between
consecutive spikes fired by a neuron. Histograms of interspike times from the rw
model and the Reg network model are shown in Fig 5C. 
While not identical, they bear a clear resemblance.

Below we propose an explicit distribution that will be shown numerically to 
approximate well the distribution of interspike times for the rw model. 
We will then apply these ideas to the network models, and see how they fare. 

\bigskip \noindent
{\bf Approximation of first passage times of the rw model by inverse Gaussians.}
For convenience, we consider a rescaling of $(X^E_t, X^I_t)$ in which
the interval $[0,M]$ is scaled linearly to  $[0,1]$, with jump sizes scaled accordingly.
(We may assume for purposes of this discussion that jump sizes are given
by $S_{QQ'}/M$ whether or not $S_{QQ'}$ is an integer.)
Let us call this rescaled rw model $(Y^E_t, Y^I_t)$,
and assume throughout that the population firing rates $(f_E, f_I)$
are those obtained from
Theorem 5.1. Let us also agree to ignore the time spent in refractory, 
which is entirely irrelevant in this discussion.

The random variables of interest, then, are $T^{rw}_E$ and $T^{rw}_I$, 
the first passage times of $Y^E_t$
and $Y^I_t$ to $1$ starting from $Y^E_0, Y^I_0=0$. For definiteness,
we will work with $Y^E_t$; the analysis of $Y^I_t$ is entirely analogous.
Below we make a sequence of approximations that will result in 
an explicit distribution to be compared to that of $T^{rw}_E$.
 
\bigskip \noindent
(i) For a small time interval $\mathrm{d} t$, we have
$$
  Y^{E}_{t + \mathrm{d} t} \approx Y^{E}_{t} + G^{E}(Y^{E}_{t}, \mathrm{d}t) \,,
$$
where
$$
  G^{E}(Y_{t}, \mathrm{d}t) = \frac{S_{EE}}{M} \mathrm{Pois}(N_{E}f_{E}P_{EE}
  \mathrm{d}t) - \frac{S_{EI}}{M}(MY^{E}_{t}) \mathrm{Pois}(N_{I}f_{I}P_{EI}
  \mathrm{d}t) + \frac{1}{M} \mathrm{Pois}( \lambda^{E} \mathrm{d}t) \,,
$$
where $\mathrm{Pois}(\lambda)$ is a Poisson random variable with
parameter $\lambda$. All Poisson random variables are assumed to be
independent. It is easy to see that
$$
  \mathbb{E}[ G^{E}( Y^{E}_{t}, \mathrm{d}t)] = \frac1M 
  \left(S_{EE}N_{E} f_{E} P_{EE} -
  S_{EI}(M Y^{E}_{t})N_{I}f_{I}P_{EI} + \lambda^{E} \right) \mathrm{d}t := b_{E}(Y^{E}_{t})\mathrm{d}t
$$
and
\begin{eqnarray*}
  \mathrm{Var}[ G^{E}( Y^{E}_{t}, \mathrm{d}t)] &=& \frac{1}{M^2}\left(S_{EE}^{2} N_{E} f_{E} P_{EE} +
  S_{EI}(M Y^{E}_{t})^{2} N_{I}f_{I}P_{EI} +                                                     \lambda^{E}\right) \mathrm{d}t \\
&:=&
\sigma^2_{E}(Y^{E}_{t}) \mathrm{d}t\,.
\end{eqnarray*}

\medskip \noindent
(ii) Next we approximate $G^{E}(Y^{E}_{t}, \mathrm{d}t)$ by a random variable 
$\hat{G}^{E}(\mathrm{d}t)$ that is independent of $Y^E_t$. Specifically we seek
$\hat{G}^{E}(\mathrm{d}t)$ with the property that
$$
  \mathbb{E}[ \hat{G}^{E}(\mathrm{d}t)] = f_E \mathrm{d}t \qquad \mbox{and}
\qquad \mathrm{Var}[\hat{G}^{E}(\mathrm{d}t)] \approx \mathrm{Var}[ G^{E}(
  1/2, \mathrm{d}t)] := \hat{\sigma}_{E}^{2} \mathrm{d}t\ ,
$$
i.e.,
$$
  \hat{\sigma}_{E} = \frac{1}{M}\sqrt{S^{2}_{EE}N_{E}f_{E}P_{EE} +
    S_{EI}(M/2)^{2}N_{I}f_{I}P_{EI} + \lambda^{E} } \,.
$$

We leave it to the reader to check that the following might be a candidate:
\begin{eqnarray*}
  \hat{G}^{E}( \mathrm{d}t) &=& (\frac{S_{EE}}{M} + \epsilon)\mathrm{Pois}(N_{E}f_{E}P_{EE}
  \mathrm{d}t) - ( \frac{S_{EI}(M/2)}{M} - \epsilon) \mathrm{Pois}(N_{I}f_{I}P_{EI}
  \mathrm{d}t) \\
&&+ (\frac{1}{M }+ \epsilon) \mathrm{Pois}( \lambda^{E} \mathrm{d}t)  \,
\end{eqnarray*}
where
$$
\epsilon = \frac{f_{E} - 
  b_{E}(1/2)}{M(N_{E}f_{E}P_{EE} + N_{I}f_{I}P_{EI} +
  \lambda^{E})} \,.
$$

\medskip \noindent
(iii) It is well known that a Poisson distribution $\mathrm{Pois}(\lambda)$ is
approximated by $N(\lambda, \lambda)$ where $N(\cdot, \cdot)$ is the
normal distribution when $\lambda$ is large (usually
larger than $10$). In our model, the three constants $N_{E}f_{E}P_{EE}$, 
$N_{I}f_{I}P_{EI}$, and $\lambda^{E}$ are $>10^{3}$ for strong drive. 
Under these conditions, for $\mathrm{d}t>0.01$, the three Poisson distributions 
can be approximated by
normal distributions. Since a linear combination of independent normal
random variables gives a normal random variable, we have the approximation
$$
  Y^{E}_{t + \mathrm{d}t} \approx Y^{E}_{t} + f_E \mathrm{d}t +
  \hat{\sigma}_{E} \mathrm{d} W_{ \mathrm{d}t} \,,
$$
where $\mathrm{d}W_{\mathrm{d}t} \sim N(0, \mathrm{d}t)$. 

\medskip \noindent
(iv) The formula above is the Euler-Maruyama scheme for
the stochastic differential equation
\begin{equation}
\label{sde}
  \mathrm{d}Z_{t} = f_{E} \mathrm{d}t + \hat{\sigma}_{E}
  \mathrm{d}W_{t}\quad , \quad
  Z_{0} = 0 \,.
\end{equation}
This  scheme is known to be strongly convergent, i.e., trajectories produced 
by the numerical scheme
converges to trajectories of \eqref{sde} when the step size approaches
to $0$ \cite{kloeden2013numerical}.  

\medskip \noindent
(v) Finally, for a (true) Brownian motion with a drift, given by
$$
f_E t + \hat \sigma_E W_t\ ,
$$
the first passage time to $1$ starting from $0$ is given by the inverse Gaussian  
$IG(f_E^{-1}, {\hat \sigma}_{E}^{-2})$ \cite{chhikara1988inverse}, where the inverse
Gaussian $IG(\mu, \nu)$ is the probability distribution with density
$$
\rho(x; \mu, \nu) = \left[\frac{\nu}{2\pi x^3}\right]^\frac12 \mbox{ exp}
\left\{\frac{-\nu(x-\mu)^2}{2\mu^2 x}\right\}\ .
$$

\medskip
We remark that we do not claim to have control over 
the cumulative errors in steps (ii), (iii) and (iv), and that the argument above is 
intended only to be heuristic. 
Numerically, it appears to be a good approximation,
as can be seen in Fig 6A, where we have plotted the first passage
times for $Y^E_t$ and $Y^I_t$ and their inverse Gaussian
approximations.

\bigskip \noindent
{\bf Interspike times in network models and inverse Gaussians.}
The idea here is to start from a network model, pass to its accompanying 
rw model, find the appropriate IG distribution as discussed above, and 
use that to approximate the interspike times of the network model. 
The excellent match between inverse Gaussians and the pdf of interspike times
for the Hom model can be seen in Fig 6B, both in background and under drive.
The performance of this recipe is likely to deteriorate with increased synchronization:
not only will synchronization change firing rate (hence the drift term in the approximation above) as discussed in Sect. 4.2, in the case of strongly synched
networks the semi-regularity of spiking events will likely appear as semi-regularly
spaced bumps in interspike time distributions. 
Still, in lower resolution the inverse Gaussian
appears to be a decent approximation  for the Reg network, as can
be seen in Fig 6C.

\begin{figure}[htbp]
\centerline{\includegraphics[width = \linewidth]{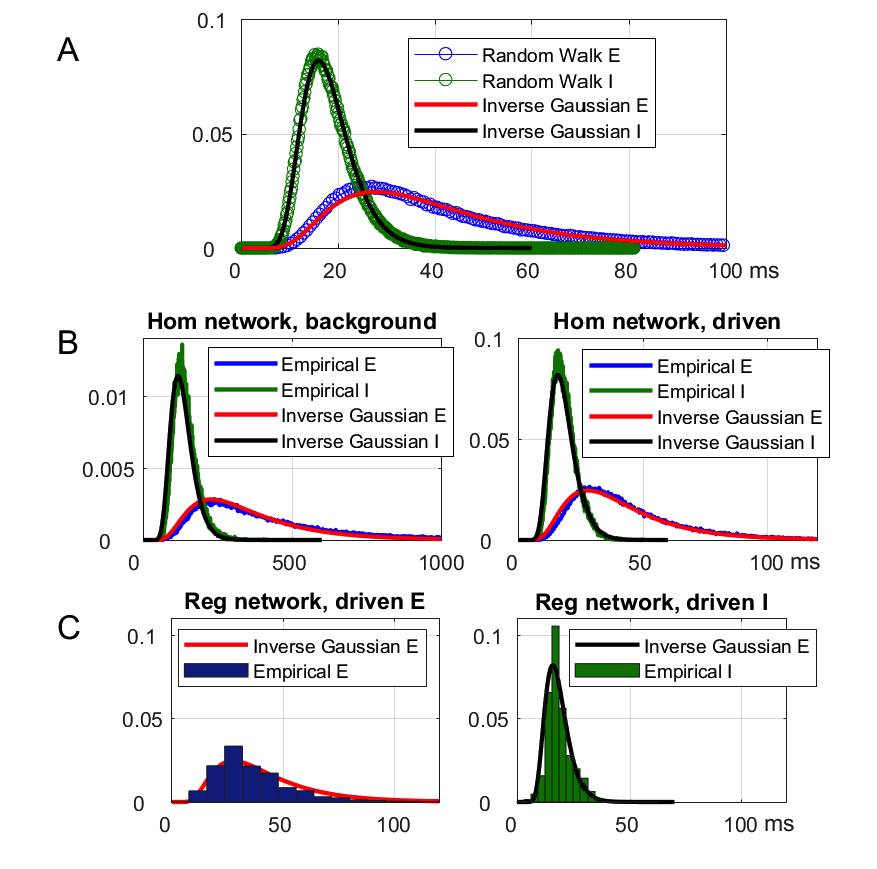}}
\caption{{\bf Inverse Gaussians as approximations for distributions of interspike
times.} A. Pdf of the inverse Gaussian distributions (solid lines, red for E, black for I)
and empirical first passage time for $Y^{E}_{t}$ and $Y^{I}_{t}$
(open circles, blue for E, green for I). For parameters of the IG-distributions, see
the main text.
B. Comparison of inverse Gaussian distribution (red/black) and and empirical
interspike times from the Hom network (blue/green). Left: background.
Right: strong drive. Resolution: 1000 bins. C. Comparison of inverse
Gaussian distribution (red/black) and and empirical interspike
times from the Reg network (blue/green) at much lower resolutions. 
Left: E-interspike times. Right: I-interspike times. 
bins.}
\label{fig6}
\end{figure}


\bigskip
\section{Summary and conclusion}

We introduced in Section 1 a family of stochastic networks
of interacting neurons that we believe will be of independent interest, but for
our purposes here, what is relevant is that these models have easily characterizable
firing rates and correlation properties, which emerge as a result of the dynamical
interaction among neurons.

We compared the firing rates of these network models to some very simple reduced
models defined by a pair of ODEs representing 
the membrane potentials of
an E and an I-neuron, taking care to give these neurons the mean excitatory and inhibitory currents received per unit time by E and I-neurons in 
the stochastic network models.

A property common to many reduced models including ours
is the underlying assumption that all inputs received by a neuron arrive
in a time-homogeneous way. This assumption contradicts directly 
the presence of correlations in neuronal spiking, which are
observed in network models as in the real brain. It is arguably the single
biggest difference between reduced and network models.

How exactly does the uneven arrival of input affect firing rate? 
Does the Ergodic Theorem not tell us that in the long run, the integral of 
the net current, roughly speaking, determines firing rate? The answer would
have been yes had it not been for the ``nonlinearities" present in the time 
evolution of membrane potentials. We have focused on 
two of these nonlinearities: refractory and the $V$-dependence of I-currents,
and our findings can be summarized as follows:

A neuron ``misses" a fraction of the spikes that it receives during refractory. 
The fraction missed is proportional to firing rate, so for a start 
refractory affects E and I neurons differently. The amount of current missed can be significant, and is exaggerated by synchronization. Furthermore, the composition 
of the missed current matters, and that has to do with 
relative spike times and conductance properties of E and I neurons.
Synchronization also changes of the membrane potential of the postsynaptic 
neuron when I-spikes arrive, altering the strength of the suppression.

We demonstrated clearly how these nonlinearities affected firing rates in 
three example networks, using simulations to confirm the reasoning above. In these and other examples we have studied, 
the effects can be considered modulatory except when the network 
is highly synchronized, and the errors caused by the two nonlinearities above 
can add or cancel. We believe these pictures are quite typical, but obviously cannot
draw more general conclusions.

In addition to the use of simple ODE models to predict firing rates, we have
found that simple random walk models (again using currents
from the network model to determine step sizes and biases) 
reproduce well fluctuations in membrane potentials 
and inter-spike time distributions of individual neurons in network models.

\bigskip
\bibliography{myref}
\bibliographystyle{amsplain}
\end{document}